\documentclass{article}

\usepackage{amsmath,amssymb,fullpage,times,amsthm,color}
\usepackage[final]{graphicx}

\usepackage{algorithm}
\usepackage{algpseudocode}

\def\part{{\color{black}{part}}}

\newtheorem*{sandwichability}{Sandwichability}

\newtheorem*{optimizer}{Entropic Optimizer}
\newtheorem*{convexity}{Convexity}

\newtheorem*{epentropy}{Edge Profile Entropy}
\newtheorem*{edgeblocksymmetry}{Edge Block Symmetry}

\newtheorem{theorem}{Theorem}
\newtheorem{proposition}{Proposition}
\newtheorem{definition}{Definition}
\newtheorem{lemma}{Lemma}

\newtheorem{remark}{Remark}

\def\N{\mathrm{I\!N}}
\def\R{\mathrm{I\!R}}

\def\P{\mathbb{P}}

\def\erdos{Erd\H{o}s}
\def\renyi{R\'{e}nyi}

\def\ENT{{\sc Ent}}
\def\MU{Thickness}

\newcommand{\sand}{sandwichable}

\def\P{\mathbb{P}}

\newcommand{\nats}{\mathbb{N}}

\begin{document}

\title{Product Measure Approximation of Symmetric Graph Properties}

\author{
Dimitris Achlioptas
\thanks{Research supported by a European Research Council (ERC) Starting Grant (StG-210743) and an Alfred P. Sloan Fellowship.} 
\\ Department of Computer Science\\ University of California, Santa Cruz\\{\tt optas@cs.ucsc.edu}\\
\and 
Paris Siminelakis
\thanks{Supported in part by an Onassis Foundation Scholarship.}
 \\ Department of Electrical Engineering\\ Stanford University\\{\tt psimin@stanford.edu}
}

\date{\empty}

\maketitle

\begin{abstract}

In the study of random structures we often face a trade-off between realism and tractability, the latter typically enabled by assuming some form of independence. 
In this work we initiate an effort to bridge this gap by developing tools that allow us to work with independence without assuming it. Let $\mathcal{G}_{n}$ be the set of all graphs on $n$ vertices and let $S$ be an arbitrary subset of $\mathcal{G}_{n}$, e.g., the set of graphs with $m$ edges. The study of random networks can be seen as the study of properties that are true for \emph{most} elements of $S$, i.e., that are true with high probability for a uniformly random element of $S$. With this in mind, we pursue the following question: \emph{What are general sufficient conditions for the uniform measure on a set of graphs $S \subseteq \mathcal{G}_{n}$ to be approximable by a product measure?}
\end{abstract}

\section{Introduction}

Since their introduction in 1959 by \erdos\ and \renyi~\cite{erdHos1959} and Gilbert~\cite{gilbert}, respectively, $G(n,m)$ and $G(n,p)$ random graphs have dominated the mathematical study of random networks~\cite{bolo_book,Janson_book}. Given $n$ vertices, $G(n,m)$ selects uniformly among all graphs with $m$ edges, whereas $G(n,p)$ includes each edge independently with probability $p$. A refinement of $G(n,m)$ are graphs chosen uniformly among all graphs with a given degree sequence, a distribution made tractable by the configuration model of Bollob{\'a}s~\cite{bolo_book}. Due to their mathematical tractability these three models have become a cornerstone of Probabilistic Combinatorics 
and have found application in 
the Analysis of Algorithms, 
Coding Theory, 
Economics, 
Game Theory, 
and Statistical Physics.

This mathematical tractability stems from symmetry: the probability of each edge is either the same, as in $G(n,p)$ and $G(n,m)$, or merely a function of the potency of its endpoints, as in the configuration model. This extreme symmetry bestows such graphs with numerous otherworldly properties such as near-optimal expansion. Perhaps most importantly, it amounts to a complete lack of geometry, as manifest by the fact that the shortest path metric of such graphs suffers maximal distortion when embedded in Euclidean space~\cite{distortion}. In contrast, vertices of real networks are typically embedded in some low-dimensional geometry, either explicit (physical networks), or implicit (social and other latent semantics networks), with distance being a strong factor in determining the probability of edge formation.

While these shortcomings of the classical models have long been recognized, proposing more realistic models is not an easy task. The difficulty lies in achieving a balance between realism and analytical tractability: it is only too easy to create network models that are both ad hoc and intractable. By now there are thousands of papers proposing different ways to generate graphs with desirable properties~\cite{Goldenberg} and the vast majority of them only provide heuristic arguments to back up their claims. For a gentle introduction the reader is referred to the book of Newman~\cite{newman2010networks} and for a more mathematical treatment to the books of Chung and Lu~\cite{chung2006complex} and of Durrett~\cite{durrett}. 

In trying to replicate real networks one approach is to keep adding features, creating increasingly complicated models, in the hope of  \emph{matching} observed properties. Ultimately, though, the purpose of any good model is prediction. In that sense, the reason to study random graphs with certain properties is to understand what \emph{other} graph properties are \emph{typically} implied by the assumed properties. For instance, the reason we study the uniform measure on graphs with $m$ edges, i.e., $G(n,m)$, is to understand ``what properties are typically implied by the property of having $m$ edges" (the answer cast as  ``properties that hold with high probability in a `random' graph with $m$ edges"). Notably, analyzing the uniform measure even for this simplest property is non-trivial. The reason is that it entails the single massive random choice of an $m$-subset of edges, rather than many independent choices. In contrast, the independence of edges in $G(n,p)$ makes the distribution far more accessible, dramatically enabling analysis.

It is a classic result of random graph theory that 
for $p=p(m)=m/\binom{n}{2}$, the random graph $G\sim G(n,m)$ and the two random graphs $G^{\pm} \sim G(n,(1\pm\epsilon)p)$ can be coupled so that, viewing each graph as a set of edges, with high probability,
\begin{equation}
\label{eq-sandwich}
G^{-} \subseteq G \subseteq G^{+} \enspace .
\end{equation}
The significance of this relationship between what we \emph{wish} to study (uniform measure) and what we \emph{can} study (product measure) can not be overestimated. It manifests most dramatically in the study of monotone properties: to study such a property in $G\sim G(n,m)$, it suffices to consider $G^+$ and show that it typically does not have the property (negative side), or $G^-$ and show that it typically does have the property (positive side). This connection has been thoroughly exploited to establish threshold functions for a host of monotone graph properties such as Connectivity, Hamiltonicity, and Subgraph Existence, making it the workhorse of random graph theory. 

In this work we seek to extend the above relationship between the uniform measure and product measures to properties more delicate  than having a given number of edges. In doing so we (i) provide a tool that can be used to revisit a number of questions in random graph theory from a more realistic angle and (ii) lay the foundation for designing random graph models eschewing independence assumptions. For example, our tool makes short work of the following set of questions (which germinated our work):\smallskip

Given an arbitrary collection of $n$ points on the plane what can be said about the set of all graphs that can be built on them using a given amount of wire, i.e., when connecting two points consumes wire equal to their distance? What does a uniformly random such graph look like? How does it change as a function of the available wire?

\subsection{Our Contribution}

A product measure on the set, $\mathcal{G}_{n}$, of all undirected simple graphs with $n$ vertices is specified by a symmetric matrix $\mathbf{Q} \in [0,1]^{n\times n}$ where $Q_{ii}=0$ for $i\in[n]$. By analogy to $G(n,p)$, we denote by $G(n,\mathbf{Q})$ the measure in which each possible edge $\{i,j\}$ is included independently with probability $Q_{ij}=Q_{ji}$. Let $S \subseteq \mathcal{G}_{n}$ be arbitrary. Our main result is a sufficient condition for the uniform measure over $S$, denoted by $U(S)$, to be approximable  by a product measure in the following sense.
\begin{sandwichability}
The measure $U(S)$  is $(\epsilon,\delta)$-\sand\ if there exists an $n\times n$ symmetric matrix $\mathbf{Q}$ such that the distributions $G^{\pm} \sim G(n,(1\pm\epsilon)\mathbf{Q})$, and $G \sim U(S)$ can be coupled so  that $G^{-} \subseteq G \subseteq G^{+}$ with probability at least $1-\delta$.
\end{sandwichability}

Informally, the two conditions required for our theorem to hold are as follows:\medskip

\noindent\emph{Partition Symmetry.} The set $S$ should be symmetric with respect to \emph{some} partition of the $\binom{n}{2}$ possible edges. That is, the characteristic function of $S$ can depend on \emph{how many} edges are included from each part but not on \emph{which}. The set of all graphs with $m$ edges satisfies this trivially: place all edges in one part and let the characteristic function be the indicator that exactly $m$ edges are included. Far more interestingly, in our motivating example edges are partitioned into equivalence classes according to their length (distance of endpoints) and the characteristic function allows every combination (vectors) of number of edges from each part that does not violate the total wire budget. We discuss the motivation for edge-partition symmetry at length in Section~\ref{symmetry}.\smallskip 

\noindent\emph{Convexity.} Partition symmetry  reduces the study of $U(S)$ to the study of the induced measure on the set of all possible combinations (vectors) of numbers of edges from each part. As, in principle, this set can be arbitrary we must impose some regularity. We chose to do this by requiring this discrete set to be \emph{convex} (more precisely that it contains all integral points in its convex hull). While convexity is actually not necessary for our proof method to work it provides a clean conceptual framework while allowing very general properties to be expressed. These include properties expressible as Linear Programs over the number of edges from each part and even non-linear sets of constraints expressing the presence or absence of percolation. Most importantly, since convex sets are closed under intersection, convex properties can, be composed (set intersection) while maintaining approximability by a product measure. \medskip

We state our results formally in Section~\ref{sec:results}. The general idea is this.  
\begin{theorem}[Informal] 
If $S$ is a convex symmetric set with sufficient edge density, then $U(S)$ is sandwichable by a product measure $G(n,\mathbf{Q}^{*})$.
\end{theorem}
The matrix $\mathbf{Q}$ is constructed by maximizing a concave function (entropy) over a convex domain (the convex hull mentioned above). As a result, in many cases it can be computed explicitly, either analytically or numerically. Further, it essentially characterizes the set $S$, as all quantitative requirements of our theorem are expressed only in terms of $\mathbf{Q}^{*}$, $k$ the number of parts in the partition and $n$ the number of vertices. 

The proof of the theorem crucially relies on a new concentration inequality we develop for symmetric subsets of the binary cube and which, as we shall see, is \emph{sharp}. Besides enabling the study of monotone properties, our results allow obtaining tight estimates of moments (expectation, variance, etc) of local features like subgraph counts.

\paragraph{Outline of the Paper}

In the next two sections we provide motivation (Section~\ref{symmetry}) and some example applications of our theorem (Section~\ref{sec:examples}). We state our results formally in  Section 4 and provide a technical overview of the proofs in Section 5. In Section 6, we discuss a connection of our work to Machine Learning and to Probabilistic Modeling in general. The rest of the  paper is devoted to the proofs. In Sections~\ref{sec:proof} and \ref{sec:sandwich} we present the proofs of concentration and sandwichability respectively. Finally,  Section~\ref{app:basic} provides the proofs of some basic supplementary results.	

\section{Motivation}\label{symmetry}

As stated, our goal is to enable the study of the uniform measure over sets of graphs. The first step in this direction is to identify a ``language" for specifying sets of graphs that is expressive enough to be interesting and restricted enough to be tractable.

Arguably the most natural way to introduce structure on a set is to impose symmetry. This is formally expressed as the \emph{invariance} of the set's characteristic function under the action of a group of transformations. In this work, we explore the progress that can be made if we define an \emph{arbitrary} partition of the edges and take the set of transformations to be the the Cartesian product of all possible permutations of the edges (indices) within each part (symmetric group). While our work is only a first step towards a theory of extracting independence from symmetry, we argue that symmetry with respect to an edge partition is well-motivated for two reasons. \smallskip

\noindent{\bf Existing Models.} 
The first is that such symmetry, typically in a very rigid form, is already implicit in several random graph models besides $G(n,m)$. Among them are Stochastic Block Models (SBM), which assume the much stronger property of symmetry with respect to a \emph{vertex} partition, and Stochastic Kronecker Graphs~\cite{leskovec2010kronecker}. The fact that our notion of symmetry encompasses SBMs is particularly pertinent in light of the theory of Graph Limits~\cite{Lovasz2012}, since inherent in the construction of the limiting object is an intermediate approximation of the sequence of graphs by a sequence of SBMs, via the (weak) Szemer\'{e}di Regurality Lemma~\cite{frieze1999quick,borgs2014p}. Thus, any property that is encoded in the limiting object, typically subgraph densities, is expressible within our framework.\smallskip

\noindent{\bf Enabling the Expression of Geometry.} 
A strong driving force behind the development of recent random graph models has been the incorporation of geometry, an extremely natural backdrop for network formation. Typically this is done by embedding the vertices in some (low-dimensional) metric space and assigning probabilities to edges as a function of length. Our work enables a far more light-handed approach to incorporating geometry by viewing it (i) as a symmetry rendering edges of the same length equivalent, while (ii) recognizing that it imposes \emph{macroscopic} constraints on the set of feasible graphs. Most obviously, in a physical network where edges (wire, roads) correspond to a resource (copper, concrete) there is a bound on how much can be invested to create the network while, more generally, cost (length) may represent a number of different notions (e.g., class membership) that distinguish between edges. 

Perhaps the \emph{most significant feature} of our work is that it fully supports the expression of geometry, by allowing the partition of edges into equivalence classes, without imposing any specific geometric requirement, i.e., without mandating the partition. \medskip

\noindent{\bf Why Convexity?} As mentioned, partition symmetry  reduces the study of $U(S)$ to the study of the induced distribution on valid combinations (vectors) of numbers of edges from each part. Without any assumptions this set can be arbitrary, e.g., $S$ can be the set of graphs having either $n^{1/2}$ or $n^{3/2}$ edges, rendering any approximation by a product measure hopeless. To focus on the cases where an approximation would be meaningful we adopt the conceptually clear and intuitive condition of convexity. In reality, convexity of the set is actually a proxy of the real property that we use in the proof, that of \emph{approximate unimodality} (Section \ref{sec:technical}).

\section{Applications}\label{sec:examples}

Given a partition, let $m_{i}(G)$ denote the number of edges from part $i$ in a graph $G$, let $\mathbf{m}(G)$ denote the \emph{edge-profile} of $G$, and let $\mathbf{m}(S)=\{\mathbf{m}(G): G \in S\}$ for $S \subseteq \mathcal{G}_n$.\smallskip

\noindent {\bf Bounded Budget.} Given a vertex set $V$ and any cost function $c:V\times V\to \R_{+}$ partition the edges into equivalence classes according to cost. Given a budget $B$, let $S(B)=\{G\in \mathcal{G}_{n}| \sum_{i}m_{i}(G)c_{i} \leq B\}$, i.e, the set of all graphs feasible with budget $B$. Using the tools developed in this paper, we can study the uniform measure on $S(B)$ and show that the probability $Q_{uv}^{*}$ of an edge $\{u,v\}$ with cost $c_{uv}$ is exponentially small in its cost, specifically $\displaystyle Q_{uv}^{*} = [1+\exp(\lambda(B) c_{uv})]^{-1}$, where $\lambda(B)$ is decreasing in $B$.\smallskip

\noindent{\bf Linear Programs.} Instead of a single cost for each edge and a single budget, edges can have multiple attributes, e.g., throughput, latency, etc. Grouping edges with identical attributes in one part, we can write arbitrary linear systems whose variables are the components of the edge-profile $\mathbf{m}$, expressing capacity constraints, latency constraints, explicit upper and lower bounds on the acceptable number of edges from a class, etc.  Now, $S=\{G\in \mathcal{G}_{n}| A\cdot \mathbf{m}(G)\leq \mathbf{b}  \}$, 
for some matrix $\pmb{A}=[\pmb{A}_{1}\ldots \pmb{A}_{k}]\in \R^{\ell \times k}$ and vector $\mathbf{b}\in \R^{\ell}$. Besides generality of expression, the entropy optimization problem defining $\mathbf{Q}$ has a closed form analytic solution in terms of the dual variables $\pmb{\lambda}\in \R_{+}^{\ell}$. 
The probability of an edge $(u,v)$ in part $i$ is now given by:
$Q_{uv}^{*}(S) = \left[1+\exp(\pmb{A}_{i}^{T}\pmb{\lambda})\right]^{-1}$. In Section~\ref{sec:conclavio}, we show how this result can be used to justify the assumptions of Logistic Regression.\smallskip

\noindent{\bf Navigability.} 
Kleinberg~\cite{Nature,NIPS} gave sufficient conditions for greedy routing to discover paths of poly-logarithmic length between any two vertices in a graph. One of the most general settings where such navigability is possible is set-systems, a mathematical abstraction of the relevant geometric properties of grids, regular-trees and graphs of bounded doubling dimension. The essence of navigability lies in the requirement that for any vertex in the graph, the probability of having an edge to a vertex at distance in the range $[2^{i-1},2^{i})$ is approximately uniform for all $i\in[\log n]$. 

In our setting, we can partition the $\binom{n}{2}$ edges according to distance scale so that part $P_{i}$ includes all possible edges between vertices at scale $i$. By considering graphs of bounded budget as above, in~\cite{navigating} we recover Kleinberg's results on navigability in set-systems, but without any independence assumptions regarding network formation, or coordination between the vertices (such as using the same probability distribution). Besides establishing the robustness of navigability, eschewing a specific mechanism for (navigable) network formation allows us to recast navigability as a property of networks brought about by economical (budget) and technological (cost) advancements. \smallskip

\noindent{\bf Percolation Avoidance.} Consider a social network consisting of $\ell$ groups of sizes $\rho_{i}n$, where $\rho_i >0$ for $i\in[\ell]$. Imagine that we require a specific group $s$ to act as the ``connector", i.e., that the graph induced by the remaining groups should have no component of size greater than $\epsilon n$ for some arbitrarily small $\epsilon > 0$. 

To study the uniform measure on the set of all such graphs  $S=S_{\epsilon}$, as with SBMs, it is natural to partition the possible edges in $\binom{\ell}{2}$ equivalence classes based on the communities of the endpoints. While the set $S_{\epsilon}$ is not symmetric with respect to this partition our result can still be useful in the following manner. Using a well known connection between Multitype Branching Process~\cite{mode1971multitype} and the existence of Giant Component in mean-field models, like Erdos-Renyi and Stochastic Block Models, we can recast the non-existence of the giant component in terms of a condition on the number of edges between each block. In the sense that conditional on the number of edges, if the condition holds  with high probability there is no giant component. Concretely, 
given the edge-profile $\mathbf{M}$ and for a given cluster $s\in [\ell]$, define the $(\ell-1)\times(\ell-1)$ matrix: \[
T(\mathbf{M})_{ij} :=\frac{m_{ij}}{n^{2}\rho_{i}}, \forall i,j \in [\ell]\setminus\{s\}
\]
that encapsulates the dynamics of a multi-type branching process. Let $\lVert \cdot \rVert_{2}$ denote the \emph{operator norm} (maximum singular value). A classic result of branching processes asserts that if $\lVert T(\mathbf{M}) \rVert_{2}< 1$ no giant component exists. Thus, in our framework, the property $S_{\epsilon}=\{ \text{ no giant component without vertices from  } s \}$, not only can be accurately approximated under the specific partition $\mathcal{P}$ by a set $\hat{S}=\left\{ \mathbf{M}: \ \lVert T(\mathbf{M}) \rVert_{2}< 1 \right\}$
that happens to also be a convex function of $\mathbf{M}$. 
\section{Definitions and Results}\label{sec:results}
For each convex symmetric set $S$ we will identify an approximating product measure $\mathbf{Q}^{*}(S)$ as the solution of a constrained entropy-maximization problem. 
A  precise version of our theorem requires the definition of parameters capturing the geometry of the convex domain along with certain aspects of the edge-partition. As we will see, our results are sharp with respect to these parameters. Below we give a series of definitions concluding with a formal statement of our main result

We start with some notation. We will use lower case boldface letters to denote vectors and uppercase boldface letters to denote  matrices. Further, we fix an arbitrary enumeration of the $N=\binom{n}{2}$ edges and sometimes represent the set of all graphs on $n$ vertices as $H_{N}=\{0,1\}^{N}$.  
We will refer to an element of $x\in H_{N}$ interchangeably as a graph and a string.  
Given a partition $\mathcal{P}=(P_{1},\ldots,P_{k})$ of $[N]$, we define $\Pi_{N}(\mathcal{P})$ to be the set of all permutations acting only within blocks of the partition.
 
\begin{edgeblocksymmetry}
\label{def-symmetry}
Fix a partition $\mathcal{P}$ of $[N]$. A set $S\subseteq H_{N}$ is called $\mathcal{P}$-symmetric if it is invariant under the action of $\Pi_{N}(\mathcal{P})$. Equivalently, if $\mathbb{I}_{S}(x)$ is the indicator function of set $S$,  then $\mathbb{I}_{S}(x) = \mathbb{I}_{S}(\pi(x))$ for all $x\in H_{N}$ and $\pi \in \Pi_{N}(\mathcal{P})$.
\end{edgeblocksymmetry}
The number of blocks $k=|\mathcal{P}|$ gives a rough indication of the amount of symmetry present. For example, when $k=1$  we have maximum symmetry and all edges are equivalent. In a stochastic block model (SBM) with $\ell$ classes, $k=\binom{\ell}{2}$.
For a $d$-dimensional lattice, partitioning the potential edges by distance results in roughly $k = n^{1/d}$ parts, whereas finally if $k=N$ there is no symmetry whatsoever. Our results accommodate partitions with as many as $O(n^{1-\epsilon})$ parts. This is way more than enough for most situations. For example, as we saw, in lattices there are $O(n^{1/d})$ distances, while if we have $n$ generic points such that the nearest pair of points have distance 1 while the farthest have distance $D$, fixing any $\delta > 0$ and binning together all edges of length $[(1+\delta)^i, (1+\delta)^{i+1})$ for $i\ge 0$, yields only $O(\delta^{-1}\log D)$ classes.
 
Recall that given a partition $\mathcal{P}=(P_{1},\ldots,P_{k})$ of $H_{N}$ and a graph  $x\in H_{N}$, the edge profile of $x$ 
is $\mathbf{m}(x) := (m_{1}(x),\ldots,m_{k}(x))$, where $m_{i}(x)$ is the number of edges of $x$ from $P_{i}$, and that the image of a $\mathcal{P}$-symmetric set $S$ under $\mathbf{m}$ is denoted as $\mathbf{m}(S)\subseteq \R^{k}$. The edge-profile is crucial to the study of $\mathcal{P}$-symmetric sets due to the following basic fact (proven in Section~\ref{app:basic}).
\begin{proposition}\label{paparia}
Any function $f:H_{N}\to \R$  invariant under $\Pi_{N}(\mathcal{P})$ depends only on  the edge-profile $\mathbf{m}(\mathbf{x})$.
\end{proposition}

In particular, since membership in $S$ depends solely on a graph's edge-profile, it follows that a uniformly random element of $S$ can be selected as follows: (i) select an edge profile $\mathbf{v}=(v_{1},\ldots,v_{k})\in\mathbb{R}^{k}$ from the distribution on $\mathbf{m}(S)$  induced  by $U(S)$, and then (ii) for each $i \in [k]$ independently select a uniformly random $v_i$-subset of $P_i$. (Formally, this is Proposition~\ref{paparaki} in Section~\ref{app:basic}.) Thus, conditional on the edge-profile, not only is the distribution of edges known, but it factorizes in a product of $G(n,m)$ distributions. In other words, the complexity of the uniform measure on $S$ manifests entirey in the \emph{induced} distribution on $\mathbf{m}(S)\in \nats^{k}$ whose structure we need to capture.
\begin{definition}
Let $p_{i}=|P_{i}|$ denote the number of edges in \part\ $i$ of  partition $\mathcal{P}$.
\end{definition}

\begin{epentropy}
Given an edge profile $\mathbf{v} \in \mathbf{m}(S)$ define the entropy of $\mathbf{v}$ as
$\displaystyle{\text{\ENT}(\mathbf{v}) = \sum_{i=1}^{k}\log \binom{p_{i}}{v_{i}}}$.
\end{epentropy}

Using the edge-profile entropy we can express the induced distribution on $\mathbf{m}(S)$ as $\mathbb{P}(\mathbf{v})=\frac{1}{|S|}e^{\text{\ENT}(\mathbf{v})}$. The crux of our argument is now this: the only genuine obstacle to $S$ being approximable by a product measure is degeneracy, i.e., the existence of multiple, well-separated edge-profiles that maximize $\text{\ENT}(\mathbf{v})$. The reason we refer to this as degeneracy is that it typically encodes a hidden symmetry of $S$ with respect to $\mathcal{P}$. For example, imagine that $\mathcal{P} = (P_1,P_2)$, where $|P_1| = |P_2|=p$, and that $S$ contains all graphs with $p/2$ edges from $P_1$ and $p/3$ edges from $P_2$, or vice versa. Then, the presence of a single edge $e \in P_i$ in a uniformly random $G \in S$ boosts the probability of all other edges in $P_i$, rendering a product measure approximation impossible. 

Note that since $\mathbf{m}(S)$ is a discrete set, it is non-trivial to quantify what it means for the maximizer of $\text{\ENT}$ to be ``sufficiently unique". For example, what happens if there is a unique maximizer of $\text{\ENT}(\mathbf{v})$ strictly speaking,  but sufficiently many near-maximizers to potentially receive, in aggregate, a majority of the measure? To strike a balance between conceptual clarity and generality we focus on the following.

\begin{convexity}\label{def-convex}
Let $\mathrm{Conv}(A)$ denote the convex hull of a set $A$. Say that a $\mathcal{P}$-symmetric set $S\subseteq \mathcal{G}_{N}$ is convex iff the convex hull of $\mathbf{m}(S)$ contains no new integer points, i.e., if $\mathrm{Conv}(\mathbf{m}(S))\cap \nats^{k} = \mathbf{m}(S)$.
\end{convexity}
Let $H_{\mathcal{P}}(\mathbf{v})$ be the approximation to $\text{\ENT}(\mathbf{v})$ that results by replacing each binomial term with its binary entropy approximation via the first term in Stirling's approximation (see \eqref{ent-prof} in Section~\ref{sec:proof}).

\begin{optimizer}
Let $\mathbf{m}^{*}=\mathbf{m}^{*}(S)\in \R^{k}$ be the solution to $ \displaystyle{\max_{\mathbf{v} \in \mathrm{Conv}(\mathbf{m}(S))} H_{\mathcal{P}}(\mathbf{v})}$.
\end{optimizer}

Defining the optimization over the convex hull of $\mathbf{m}(S)$ is crucial, as it will allow us to study the set $S$ by studying only the properties of the maximizer $\mathbf{m}^{*}$. Clearly, if a $\mathcal{P}$-symmetric set $S$ has entropic optimizer $\mathbf{m}^{*} = (m^{*}_1,\ldots,m^{*}_k)$, the natural candidate product measure for each $i \in [k]$ assigns probability $m_{i}^{*}/p_{i}$ to all edges in part $P_i$. The challenge is to relate this product measure to the uniform measure on $S$ by proving concentration of the induced measure on $\mathbf{m}(S)$ around a point near $\mathbf{m}^{*}$. For that we need (i) the vector $\mathbf{m}^{*}$ to be ``close" to\footnote{Indeed, this is the only use we make of convexity in the proof presented here.} a vector in  $\mathbf{m}(S)$, and (ii) to control the decrease in entropy ``away" from $\mathbf{m}^{*}$. To quantify this second notion we need the following parameters, expressing the geometry of convex sets.
\begin{definition}
Given a partition $\mathcal{P}$ and a $\mathcal{P}$-symmetric convex set $S$, we define
\begin{align}
\text{\emph{\MU:}}\qquad 				& \mu  		=  \mu(S) 	 =	 \min_{i\in[k]} \min\{m^{*}_{i}, p_{i}-m^{*}_{i}\} \qquad \qquad \qquad\\
\text{\emph{Condition number:}} \qquad 	& \lambda 	= \lambda(S)  =    \frac{5k\log n}{\mu(S)} \label{def:cond}\\
\text{\emph{Resolution:}} \qquad 			& r 		  	= r(S)    		= 	 \frac{\lambda+\sqrt{\lambda^{2}+4\lambda}}{2}>\lambda\label{eq-resol} 
\end{align}
\end{definition}
The most important parameter is the \emph{thickness} $\mu(S)$. It quantifies the minimal  coord- inate-distance of the optimizer $\mathbf{m}^{*}(S)$ from the natural boundary $\{0,p_{1}\}\times \ldots \times \{0,p_{k}\}$ where the entropy of a class becomes zero. Thus, this parameter determines the \emph{rate of coordinate-wise concentration} around the optimum.

The \emph{condition number} $\lambda(S)$, on the other hand, quantifies the robustness of $S$. To provide intuition, in order for the product measure approximation to be accurate for every class of edges (part of $\mathcal{P}$), fluctuations in the number of edges of order $\sqrt{m_{i}^{*}}$  need to be  ``absorbed" in the mean $m_{i}^{*}$. For this to happen with polynomially high probability for a single part, standard results imply we must have $m_{i}^{*} = \Omega(\log (n))$. We absorb the dependencies between parts by taking a union bound, thus multiplying by the number of parts, yielding the numerator in~\eqref{def:cond}. Our results give strong probability bounds when $\lambda(S)\ll 1$, i.e., when in a typical graph in $S$ the number of edges from each part is $\Omega(k\log n)$ edges away from triviality, a condition we expect to hold in all natural applications. We can now state our main result.
\begin{theorem}[Main result]\label{thm-sandwich}
Let $\mathcal{P}$ be any edge-partition and let $S$ be any $\mathcal{P}$-symmetric convex set. For every $\epsilon>\sqrt{12\lambda(S)}$, the uniform measure over $S$ 
is $(\epsilon,\delta)$-\sand, where $\delta =2 \exp\left[ - \mu(S)\left(\frac{\epsilon^{2}}{12}-\lambda(S) \right)\right]$.
\end{theorem}

\begin{remark}
As a sanity check we see that as soon as  $m \gg \log n$, Theorem~\ref{thm-sandwich} recovers the sandwichability of $G(n,m)$ by $G(n,p(m))$ as sharply as the Chernoff bound, up to the constant factor $\frac{1}{12}$ in the exponent.
\end{remark}

Theorem~\ref{thm-sandwich} follows by analyzing the natural coupling between the uniform measure on $S$ and the product measure corresponding to the entropic optimizer $\mathbf{m}^{*}$. Our main technical contribution is  Theorem~\ref{thm-concentration} below, a concentration inequality for $\mathbf{m}(S)$ when $S$ is a convex symmetric set. The \emph{resolution}, $r(S)$, defined in~\eqref{eq-resol} above, reflects the smallest \emph{concentration width} that can be proved by our theorem. When $\lambda(S)\ll 1$, as required for the theorem to be meaningfully applied, it scales optimally as $\sqrt{\lambda(S)}$.

\begin{theorem}\label{thm-concentration}
Let $\mathcal{P}$ be any edge-partition, let $S$ be any $\mathcal{P}$-symmetric convex set, and let $\mathbf{m}^{*}$ be the entropic optimizer of $S$. For all $\epsilon> r(S)$, if $G\sim U(S)$, then 
\begin{equation}
\mathbb{P}_{S}\left(\left|\mathbf{m}(G)-\mathbf{m}^{*}\right|\leq \epsilon \tilde{\mathbf{m}}^{*}\right) \geq 1-  \exp\left(-\mu(S)\left(\frac{\epsilon^{2}}{1+\epsilon}-\lambda(S) \right) \right) \enspace ,
\end{equation}
where $\mathbf{x}\leq \mathbf{y}$ means that  $x_{i}\leq y_{i}$ for all $i\in [k]$, and $\tilde{m}_{i}=\min\{m^*_{i},p_{i}-m^*_{i}\}$.
\end{theorem}

The intuition driving concentration is that as \emph{thickness} increases two phenomena occur: (i) vectors close to $\mathbf{m}^{*}$ capture a larger fraction of the measure, and (ii) the decay in entropy away from $\mathbf{m}^{*}$ becomes steeper. These joint forces compete against the probability mass captured by vectors ``away" from the optimum. The point were they prevail corresponds to $\lambda(S)\ll 1$ or, equivalently, $\mu(S)\gg 5k\log(n)$. Assuming $\lambda(S)\ll 1$ the probability bounds we give scale as $n^{-\Omega(k \epsilon^{2})}$. Without  assumptions on $S$, and up to the constant 5 in~\eqref{def:cond}, this is \emph{sharp}. 

The concentration and sandwich theorems dramatically enable the study of monotone properties under the uniform measure over convex-symmetric sets. Going beyond monotone properties, we would like to enable the study of more local (involving  a few edges) events. For example, we would like to be able to make statements about moments of subgraph counts (triangle, cliques, cycles) and other local-graph functions.

\section{Technical Overview}\label{sec:technical}

In this section, we present an overview of the technical work involved in the proof of our theorems. Most of the work lies in the concentration result, Theorem~\ref{thm-concentration}.\smallskip


\noindent{\bf Concentration.} The general idea is to identify a high-probability subset $\mathcal{L}\subseteq \mathbf{m}(S)$ by integrating the probability measure around the entropy-maximizing profile $\mathbf{m}^{*}$. Since ultimately our goal is to couple the uniform measure with a product measure, we need to establish concentration for every part of the edge-partition, i.e., in every coordinate. There are two main obstacles to overcome: (i) we do not know $|S|$,  and (ii) we must integrate the measure outside  $\mathcal{L}$ while concurrently quantifying  the decrease in entropy as a function of the $L_{\infty}$ distance from the maximizer $\mathbf{m}^{*}$. Our strategy to resolve the above obstacles is:\smallskip

\emph{Size of $S$.} We bound $\log|S|$ from below by  the contribution to $\log|S|$ of the entropic optimal edge-profile $\mathbf{m}^{*}$, thus upper-bounding the probability of every $\mathbf{v} \in \mathbf{m}(S)$ as
\begin{equation}\label{eq:basic}
\log \P_{S}(\mathbf{v})= \text{\ENT}(\mathbf{v}) - \log(|S|) \leq \text{\ENT}(\mathbf{v}) - \text{\ENT}(\mathbf{m}^{*}) \enspace .
\end{equation}
This is the crucial step that opens up the opportunity of relating the probability of a vector $\mathbf{v}$ to the distance $\lVert \mathbf{v}-\mathbf{m}^{*}\rVert_{2}$ through analytic properties of entropy. The importance of this step stems from the fact that all information about $S$ resides in $\mathbf{m}^{*}$ due to our choice of solving the optimization problem on the convex hull. 
\begin{proposition}\label{prop-lower}
There exists a partition $\mathcal{P}$ with $k$ parts and a convex  $\mathcal{P}$-symmetric set $S\subseteq \mathcal{G}_{n}$ such that 
$\log(|S|) -\text{\ENT}(\mathbf{m}^{*}) = \Omega(k \log(n))$.
\end{proposition}

Proposition~\ref{prop-lower} demonstrates that unless one utilizes specific geometric properties of the set $S$ enabling integration around $\mathbf{m}^{*}$, instead of using a point-bound for  $\log |S|$, a loss of $\Omega(k \log(n))$ is unavoidable. In other words, either one makes more assumptions on $S$ besides symmetry and ``convexity", or this error term is optimal as claimed.\smallskip
 
\emph{Distance bounds:} To bound from below the rate at which entropy decays as a function of the component-wise distance from the maximizer $\mathbf{m}^{*}$, we first approximate $\text{\ENT}(\mathbf{v})$ by the corresponding binary entropy to get a smooth function. Exploiting the separability, concavity and differentiability of entropy we obtain component-wise distance bounds using a second-order Taylor approximation. At this step we also lose a cumulative factor of order $3k\log n $ stemming from Stirling approximations and the subtle point that the maximizer $\mathbf{m}^{*}$ might not be an integer point. The constant $3$ can be improved, but in light of Proposition \ref{prop-lower} this would be pointless and complicate the proof unnecessarily.\smallskip

\emph{Union bound:} Finally, we integrate the obtained bounds  outside the set of interest by showing that even if all ``bad" vectors where placed right at the boundary of the set, where the lower bound on the decay of entropy is smallest, the total probability mass would be exponentially small. The loss incurred at this step is of order $2k\log n $, since there are at most $n^{2k}$ bad vectors.\smallskip

\emph{Relaxing Conclusions.} Our  theorem seeks to provide concentration simultaneously for all parts. That motivates the definition of thickness parameter $\mu(S)$ as the minimum distance  from the natural boundary  of the number of edges that any part $i$ has at the optimum $\mathbf{m}^{*}$. Quantifying everything in terms of $\mu(S)$ is a very conservative requirement. For instance, if we define the set $S$ to have no edges in a particular part of the partition, then $\mu(S)$ is $0$ and our conclusions become vacuous. Our proofs in reality generalize, to the case where we confine our attention only to a subset $I\subseteq [k]$ of blocks in the partition.  In particular, if one defines $I^{*}$ as the set of parts whose  individual \emph{thickness} parameter $\tilde{m}_{i}=\min\{m_{i}, p_{i}-m_{i}\}$ is greater than $5k\log n $, both theorems hold for the subset of edges $\cup_{i\in I^{*}}P_{i}$. In essence that means that for every class that is ``well-conditioned", we can provide concentration of the number of edges and approximate monotone properties of only those parts by coupling them with product measures.\smallskip

\emph{Relaxing Convexity.} Besides partition symmetry that comprises our main premise and starting point, the second main assumption made about the structure of $S$ is \emph{convexity}. In the proof (Lemma \ref{lm:distance} in Section~\ref{sec:proof}) convexity is used only to argue that: (i) the maximizer $\mathbf{m}^{*}$ will be close to some vector in $\mathbf{m}(S)$, and (ii) that the first order term in the Taylor approximation of the entropy is always negative. However, since the optimization problem was defined on the convex hull of $\mathbf{m}(S)$, in point (ii) above we are only using  convexity of $\mathrm{Conv}(\mathbf{m}(S))$ and not of the set $S$. Thus, the essential requirement on $\mathcal{P}$-symmetric sets  is  \emph{approximate unimodality}. 

\begin{definition}
A $\mathcal{P}$-symmetric set $S$ is called $\Delta$-unimodal if the solution $\mathbf{m}^{*}$ to the entropy optimization problem defined in Section 2, satisfies:
\begin{equation}
d_{1}(\mathbf{m}^{*},S):=\min_{\mathbf{v}\in \mathbf{m}(S)}\lVert \mathbf{m}^{*}- \mathbf{v}\rVert_{1} \leq \Delta
\end{equation}  
\end{definition}

Convexity essentially implies that the set $S$ is $k$-unimodal as we need to round each of the $k$ coordinates of the solution to the optimization problem to the nearest integer. Under this assumption, all our results apply by only changing the condition number of the set to $\lambda(S)=\frac{(2\Delta + 3k)\log n }{\mu(S)}$. In this extended abstract, we opted to present our results by using the familiar notion of convexity to convey intuition on our results and postpone the presentation in full generality for the full version of the paper.\smallskip

\noindent\textbf{Coupling.} To prove Theorem~\ref{thm-sandwich} using our concentration result, we argue as follows. Conditional on the edge-profile, we can couple the generation of edges in different parts independently by a similar process  as in the $G(n,m)$ to $G(n,p)$ case. Then, using a union bound we can bound the probability that all couplings succeed given an appropriate $\mathbf{v}$. Finally, using the concentration theorem we show that sampling an appropriate edge-profile happens with high probability. 

\section{Discussion}\label{sec:conclavio}

In studying the uniform measure over symmetric sets of graphs our motivation was twofold. On one hand was the desire to explore the extent to which symmetric sets (properties) can be studied via product measure approximations (instead of assuming product form  and studying properties of the resulting ad-hoc random graph model). On the other hand, we wanted to provide a framework of random graph models that arise by constraining \emph{only the support} of the probability distribution and not making specific postulations about the finer probabilistic structure, e.g., at the level of edges. An unanticipated result of our work is a connection to Machine Learning (ML). \smallskip

\subsection{Connection to Machine Learning} In Machine Learning, a probabilistic model for the data is assumed and exploited to perform some inferential task (e.g. classification, clustering, estimation). Logistic Regression is a widely used modeling, estimation and classification tool. In the simplest setting, there are $N$ binary outcomes $Y_{1},\ldots, Y_{N}$ and each $Y_{i} \in \{0,1\}$ is associated  with a feature vector $\pmb{X}_{i}\in \R^{\ell}$. A statistical model is constructed by assuming that there exists an (unknown) vector $\pmb{\beta}\in \R^{\ell}$ such that: (i)  the $N$ outcomes are \emph{mutually independent} given $\pmb{\beta}$, and (ii) the marginal  probability of an outcome is given by $\P_{\pmb{\beta}}(Y_{i} = 1) = \left[1+\exp(\pmb{X}_{i}^{T}\cdot \pmb{\beta})\right]^{-1}$. Thus, given the outcomes and features one seeks to find $\pmb{\beta}$ by Maximum Likelihood Estimation (MLE). 

In the context of network modeling~\cite{Goldenberg}, the $N=\binom{n}{2}$ binary outcomes correspond to edges and the $N$ feature vectors can be either known or latent~\cite{hoff2002latent}. In both the Classification and Network Modeling settings, it is natural to wonder whether both assumptions (independence and parametric form of the likelihood) are realistic and not far-fetched idealizations. Our work provides, to the best of our knowledge, the \emph{first justification for both assumptions}, when the features (covariates) $\left\{\pmb{X}_{i}\right\}_{i\leq N}$ take only $k$ distinct values. In fact, let $\tilde{\pmb{X}}_{q}$ for $q\in[k]$ denote the common feature vector of outcomes from part $q$ and let $\tilde{\pmb{X}}=[\tilde{\pmb{X}}_{1}\ldots \tilde{\pmb{X}}_{k}]$. 
Our results say that  logistic regression is the product measure approximation of the uniform measure on the set $
S=S(\tilde{\pmb{X}},\pmb{c}(\pmb{\beta}))=\left\{x\in\{0,1\}^{N}| \ \tilde{\pmb{X}}\cdot \pmb{m}(x) \leq \pmb{c} \right\}$\footnote{For some $\pmb{c} \in \R^{\ell}$ such that $
\pmb{\beta}$ is the Lagrangian multiplier of an entropy optimization problem over $S$. } 
and that the \emph{logistic approximation}, amounting to assumptions (i)-(ii) above, is valid whenever $\mathbf{m}(x)$ is in each coordinate $\Omega(k\log n)$ far from triviality, i.e. there are \emph{enough positive and negative labels} from each part.
An in depth exploration of such connections, i.e., between widely used techniques in Machine Learning and the uniform measure over sets, is a promising direction that we pursue in future work.

\section{Proof of Theorem 2}\label{sec:proof}

In this section we prove Theorem~\ref{thm-concentration}. For the purposes of the proof we are going to employ, instead of $\mathbf{m}(x)$, a different parametrization  in terms of the \emph{edge-profile} $\mathbf{a}(x)=(a_{1}(x),\ldots,a_{k}(x))\in[0,1]^{k}$ where $a_{i}(x)=m_{i}(x)/p_{i}$. This will be convenient both in calculations as well as conceptually as $a_{i}(x)$ represents the effective edge density of a part $P_{i}$ in the partition. We start by approximating  the entropy of an edge-profile via the $\mathcal{P}$-entropy.

\begin{definition}
Given a partition $\mathcal{P}$, the $\mathcal{P}$-entropy is define for every  $\mathbf{a}\in[0,1]^{k}$ as
 \begin{equation}
 H_{\mathcal{P}}(\mathbf{a})=-\sum_{i=1}^{k}p_{i}\left[a_{i}\log a_{i}+(1-a_{i})\log(1-a_{i})\right]
\end{equation}
\end{definition}
\noindent The $\mathcal{P}$-entropy is simply the entropy of the product measure defined over edges through $\mathbf{a}$. We slightly abuse the notation and also define the $\mathcal{P}$-entropy in terms of the edge-profile:
 \begin{equation}
 \label{ent-prof}
 H_{\mathcal{P}}(\mathbf{v})=-\sum_{i=1}^{k}\left[v_{i}\log\left(\frac{v_{i}}{p_{i}} \right) +(p_{i}-v_{i})\log\left(\frac{p_{i}-v_{i}}{p_{i}} \right)\right]
\end{equation}
Let $\mathcal{M}_{\mathcal{P}}:=\{0,\ldots,p_{1}\} \times\ldots \times \{0,\ldots,p_{k}\}$ be the space of all possible vectors $\mathbf{m}$. In what follows we sometimes supress the dependence  of the quantities in $\mathbf{m}$ or $\mathbf{a}$ to ease the notation.
\begin{lemma} \label{lm:stirling}
Let $\mathbf{m}\in \mathcal{M}_{\mathcal{P}}$ be an edge-profile and $\mathbf{a}\in[0,1]^{k}$ be the corresponding probability profile, then:
\[
\text{\ENT}(\mathbf{m})= \mathcal{H}_{\mathcal{P}}(\mathbf{a})- \gamma(n)
\]
where $0 \leq \gamma(n) \leq k\log n$ is a term that approaches zero as $m_i$ and $p_i-m_i$ tend to infinity.
\end{lemma}
\begin{proof}
We begin by providing the first order Stirling approximation for a single term of the form $\log \binom{p_{i}}{m_{i}}$. Specifically, since $m_{i}=p_{i}a_{i}$ and by using $\log n!  =  n\log n - n + \frac{1}{2}\log n  + \theta_n$, where $\theta_n \in (0,1]$, we get:
\begin{eqnarray*}
\log \binom{p_{i}}{m_{i}}&=&\log(p_{i}!)-\log(m_{i}!)-\log((p_{i}-m_{i})!) \\
		 &=& -p_{i}\left[a_{i} \log a_{i} +(1-a_{i})\log(1- a_{i})\right] - \delta_n(a_i,p_i) \enspace ,
\end{eqnarray*}
where $0 \le \delta_n(a_i,p_i) \le \log n$. Summing the derived expression for all $i\in[k]$ gives:
\begin{eqnarray*}
\text{\ENT}(\mathbf{m})&=& - \sum_{i=1}^{k}p_{i}\left[a_{i} \log a_{i} +(1-a_{i})\log(1- a_{i})\right]  - \sum_{i=1}^{k} \delta_n(a_i,p_i) \\
							&=& H_{\mathcal{P}}(\mathbf{a})- \gamma(n) \enspace ,
\end{eqnarray*}
where $0 \le \gamma(n) \le k \log n$ .
\end{proof}

Next, using the Taylor remainder theorem and the connection with $\mathcal{P}$-entropy, we obtain geometric estimates on the decay of entropy around $\mathbf{m}^{*}$.

\begin{theorem}[Taylor Remainder Theorem]
Assume that $f$ and all its partial derivatives are differentiable at every point of an open set $S \subseteq \R^{K}$. If $\mathbf{a}, \mathbf{b} \in S$ are such that the line segment $L(\mathbf{a},\mathbf{b})\subseteq S$, then there exists a point $\mathbf{z}\in L(\mathbf{a},\mathbf{b})$ such that:
\begin{equation}
f(\mathbf{b})-f(\mathbf{a})=\nabla f(\mathbf{a})^{T}(\mathbf{b}-\mathbf{a})+\frac{1}{2}(\mathbf{b}-\mathbf{a})^{T}\nabla^{2}f(\mathbf{z})(\mathbf{b}-\mathbf{a}) \enspace .
\end{equation}
\end{theorem}

\begin{lemma}[$L_{2}$ distance bounds]
\label{lm:distance}
If $\mathbf{m}^{*}$ is the unique maximizer and $\mathbf{w}\in\mathrm{Conv}(\mathbf{m}(S))$, then
\begin{equation}\label{eq:niceone}
\text{\ENT}(\mathbf{w})-\text{\ENT}(\mathbf{m}^{*}) \leq -\sum_{i=1}^{k}\frac{(w_{i}-m^*_{i})^{2}}{\max\{\tilde{m}^*_i,\tilde{w}_i\}}+3k\log n\enspace 
\end{equation} 
where $\mathbf{m}^{*}_{i}=\min\{m_{i}^{*},p_{i}-m_{i}^{*} \}$(respectively $\tilde{w}_{i}$), denotes the \emph{thickness} of a part $i\in[k]$.
\end{lemma}
\begin{proof}
Invoking Lemma \ref{lm:stirling}, we rewrite the difference in entropy as a difference in $\mathcal{P}$-entropy, where $\mathbf{a}^{*}$ is the probability profile of the maximizer and $\mathbf{b}$ of $\mathbf{w}$:
\[
\text{\ENT}(\mathbf{w})-\text{\ENT}(\mathbf{m}^{*}) \leq H_{\mathcal{P}}(\mathbf{b})-H_{\mathcal{P}}(\mathbf{a^{*}})+3k \log n
\]
Here, we have additionally dealt with the subtle integrality issue, namely that $\mathbf{m}^{*}$ might not belong to $\mathbf{m}(S)$. Rounding the vector to the nearest integral point produces a cumulative error of at most $2k\log(n)$ in the entropy that adds to the $k\log(n)$ error coming from Stirling's approximation. Both errors can be reduced using higher order Stirling approximations but we avoid doing so since an error of order $k\log(n)$ is unavoidable due to the approximation of the partition function.

Convexity of the domain $\mathrm{Conv}(\mathbf{m}(S))$ and differentiability of the $\mathcal{P}$-entropy provide the necessary conditions to use the Taylor Remainder Theorem. Let  $\mathbf{z}$ be a point in the linear segment $L(\mathbf{a}^{*},\mathbf{b})$. We proceed with writing the  expressions  for partial derivatives of $H_{\mathcal{P}}$.
\begin{eqnarray}
\partial_{i}H_{\mathcal{P}}(\mathbf{a}^{*})&=&-p_{i}\log\left(\frac{a^{*}_{i}}{1-a^{*}_{i}}\right)\label{eq:grad}\\
\partial_{ii}^{2}H_{\mathcal{P}}(\mathbf{z})&=&-p_{i}\left(\frac{1}{1-z_{i}}+\frac{1}{z_{i}}\right)\label{eq:hess} \enspace ,
\end{eqnarray}
while $\partial_{ij}^{2}f=0$ for $i\neq j$ due to separability of the function $H_{\mathcal{P}}$. The Taylor Remainder forumla, now reads:
\begin{equation}
H_{\mathcal{P}}(\mathbf{b})-H_{\mathcal{P}}(\mathbf{a}^{*})= \nabla H_{\mathcal{P}}(\mathbf{a}^{*}) \cdot (\mathbf{b}-\mathbf{a}^{*})-\sum_{i=1}^{K}p_{i}(b_{i}-a^*_{i})^{2}\left(\frac{1}{1-z_{i}}+\frac{1}{z_{i}}\right)\label{eq:mitsakos}
\end{equation}
Since, $\mathbf{a}^{*}$ is the unique solution to the {sc MaxEnt} problem and the domain is convex, the first term in the above formula is always bounded above by zero. Otherwise, there would be a direction $\mathbf{u}$ and a small enough parameter $\epsilon>0$ such that $\mathbf{a}^{*}+\epsilon \mathbf{u}$ has greater entropy, a contradiction.
To bound the second sum from above, let $\tilde{z}_{i}=\min\{z_{i},1-z{i} \}$(expressing the fact that binary entropy is symmetric around $1/2$) and use the trivial bound $\tilde{z}_{i} \le {\max\{\tilde{a}^*_i,\tilde{b}_i\}} $. Thus,
\begin{equation}
\label{taylor}
H_{\mathcal{P}}(\mathbf{b})-H_{\mathcal{P}}(\mathbf{a}^{*})\leq -\sum_{i=1}^{K}p_{i}(b_{i}-a^*_{i})^{2}\frac{1}{\tilde{z}_{i}} \le
 -\sum_{i=1}^{K}p_{i}\frac{(b_{i}-a^*_{i})^{2}}{\max\{\tilde{a}^*_i,\tilde{b}_i\}} \enspace .
\end{equation}
Dividing and multiplying by $p_{i}$, and writing $\tilde{w}_{i}=p_{i}\tilde{b}_{i}$, $\tilde{m}^{*}_{i} = p_{i}\tilde{a}^{*}_{i}$, gives:
\begin{equation}
H_{\mathcal{P}}(\mathbf{b})-H_{\mathcal{P}}(\mathbf{a}^{*})\leq 
 -\sum_{i=1}^{K}\frac{(w_{i}-m^*_{i})^{2}}{\max\{\tilde{m}^*_i,\tilde{w}_i\}} \enspace .
\end{equation}
where $\mathbf{w}$ and $\mathbf{m}^{*}$ are the original edge profiles. We note that for most cases we have that $\tilde{z}_{i}= z_{i}$, i.e. a block is half-empty. 
\end{proof}

In preparation of performing the "union bound", we prove that:
\begin{proposition}
\label{classes}
The number of distinct edge-profiles $|\mathbf{m}(S)|$ is bounded by $|\mathcal{M}_{\mathcal{P}}| \leq e^{2k \log n}$.
\end{proposition}
\begin{proof}
Assuming that no constraint is placed upon $\mathbf{m}$ by $S$, then $\mathbf{m}(S)=\mathcal{M}_{\mathcal{P}}$. This number is equal to the product of $[p_{i}+1] \leq n^{2}$ as there are at most $\binom{n}{2}$ edges within a block. Multiplying the last bound we get the statement.
\end{proof}

Before proceeding with the proof of the concentration theorem, we repeat the definitions of the crucial parameters mentioned in the introduction. 
\begin{definition}
Given a partition $\mathcal{P}$ and a  $\mathcal{P}$-symmetric set $S$, define:
\begin{eqnarray}
\mu(S)&=& \min_{i\in[k]}\left\{ m^{*}_{i} \wedge (p_{i}-m^{*}_{i})\right\}\\
\lambda(S) &=& \frac{5k\log n}{\mu(S)}\\
r(S) & =& \frac{\lambda+\sqrt{\lambda^{2}+4\lambda}}{2}
\end{eqnarray}
the \emph{\MU}, \emph{condition number} and \emph{resolution} of the convex set $S$  and $x\wedge y:=\min\{ x, y\}$ denotes the $\min$ operation.
\end{definition}
\begin{proof}[Proof of Theorem~\ref{thm-concentration}]
Our goal is to use the developed machinery to control the probability of deviations from the optimum at scale $\epsilon>r(S)$. Define the set $\mathcal{L}_{\epsilon}(\mathbf{m}^{*})\stackrel{\vartriangle}{=}\{x\in S:\left|\mathbf{m}(x)-\mathbf{m}^{*} \right|\leq \epsilon \tilde{\mathbf{m}}^{*}\}$. We are going to show that $\P_{S}(\mathcal{L}^{c}_{\epsilon}(\mathbf{m}^{*}))\to 0$ ``exponentially" fast and thus provide localization of the edge profile within a scale $\epsilon$ for each coordinate. To that end, we write:
\begin{equation}
\P_{S}(\mathcal{L}^{c}_{\epsilon}(\mathbf{m}^{*})) = \sum_{\mathbf{w} \in \mathcal{L}^{c}_{\epsilon}(\mathbf{m}^{*})}\P_{S}(\mathbf{w}) \leq \sum_{\mathbf{w} \in \mathcal{L}^{c}_{\epsilon}(\mathbf{m}^{*})} \exp\left[ \text{\ENT}(\mathbf{m}) - \text{\ENT}(\mathbf{m}^{*}) \right]
\end{equation}
where we have used \eqref{eq:basic}, the approximation of the log-partition function by the entropy of the optimum edge-profile, and added the contribution of points outside of $\mathbf{m}(S)$. At this point, we are going to leverage the lower bound for the decay of entropy away from $\mathbf{m}^{*}$. This is done by first performing a union bound, i.e. considering that all points in $\mathcal{L}^{c}_{\epsilon}(\mathbf{m}^{*})$ are placed on the least favorable such point $\mathbf{w}^{*}$. Since, we are requiring coordinate-wise concentration, such point would differ from the optimal vector only in one-coordinate, and in particular should be the one that minimizes our lower bound.  Any such vector $\mathbf{w}\in \mathcal{L}_{\epsilon}(\mathbf{m}^{*})$, would have at least one coordinate $i\in[k]$ such that $|w_{i}-m_{i}^{*}|=\epsilon m^{*}_{i}$. By Lemma \ref{eq:niceone}, we get
\begin{equation}
\text{\ENT}(\mathbf{w})-\text{\ENT}(\mathbf{m}^{*}) \leq -\frac{\epsilon^{2} (m^{*}_{i})^{2}}{(1+\epsilon)m^{*}_{i}}+3k\log n = -\frac{\epsilon^{2}}{(1+\epsilon)}m^{*}_{i}+3k\log n \enspace 
\end{equation}
using the facts that $\max\{\tilde{m}_{i},\tilde{w}_{i}\}\leq \tilde{m}_{i}+\tilde{w}_{i}  \leq (1+\epsilon) m^{*}_{i}$. Now, by definition  the thickness $\mu(S)\leq \tilde{m}_{i}^{*}$ for all $i\in[k]$, and so a  a vector $\mathbf{w}^{*}$ that minimizes the bound is such that $
\text{\ENT}(\mathbf{w}^{*})-\text{\ENT}(\mathbf{m}^{*}) \leq  -\frac{\epsilon^{2}}{(1+\epsilon)}\mu(G)+3k\log n$. We perform the union bound by using  $|\mathcal{L}^{c}_{\epsilon}(\mathbf{m}^{*})| \leq |\mathcal{M}_{\mathcal{P}}| \leq \exp(2k\log n)$ from Proposition \ref{classes}:
\begin{eqnarray}
\P_{S}(\mathcal{L}^{c}_{\epsilon}(\mathbf{m}^{*})) &\leq& |\mathcal{L}^{c}_{\epsilon}(\mathbf{m}^{*})| \cdot \P_{S}(\mathbf{w}^{*})\\ 
&\leq& \exp\left[-\frac{\epsilon^{2}}{(1+\epsilon)}\mu(G)+5k\log n \right]\\ 
&\leq& \exp\left[ -\mu(G)\left(\frac{\epsilon^{2}}{1+\epsilon}-\frac{5k\log n}{\mu(G)}\right) \right]
\end{eqnarray}
Finally, identifying $\lambda(S)$ in the expression provides the statement. We note here that the resolution $r(S)$ is defined exactly so that the expression in the exponent is negative. The condition $\lambda(S)\ll 1$ is a requirement that makes  concentration possible in a small scale, i.e $\epsilon \ll 1$.
\end{proof}

\paragraph{Tightness.} The crucial steps in the proof are, firstly, the approximation of the log-partition function and, secondly, the $L_{2}$ distance bounds on the decay of entropy away from the optimum. Both steps are essentially optimal under general assumptions, as is shown in Proposition \ref{prop-lower}. Our proof can only improved by using higher order Stirling approximations and a more complicated integration process (incremental union bounds over $L_{\infty}$-shells) instead of the simple union bound, to reduce the error from $5k\log n$ down to possibly the minimum of $2k\log(n)$. Since, the above considerations would complicate the proof significantly and the gain is a small improvement in the constant we deem this unnecessary. 

\section{Proof of Theorem 1}\label{sec:sandwich}

In this section, we leverage the concentration theorem to prove that convex $\mathcal{P}$-symmetric sets are  $(\epsilon,\delta)$-\sand. Before presenting the proof of the theorem we state a preliminary lemma.
\subsection{Basic Coupling Lemma}
Consider a set of random variables $X_{1},\ldots,X_{k}$ with laws $\mu_{1},\ldots,\mu_{k}$.
A coupling between a set of random variables is a (joint) probability distribution  $\mu$, such that $\P_{\mu}(X_{i})=\P_{\mu_{i}}(X_{i})$ for $i\in[k]$, i.e. the marginals of the random variables are right. Let $A=\{1,\ldots,N\}$ be a finite set with $N$ elements.  Further let $X$ denote a uniform subset of $m$ elements of $A$, denoted as  $X\sim \mathrm{Samp}(m,A)$, and $Z$  a  subset of $A$ where each element of $A$ is included with the same probability $p$, denote as $Z\sim \mathrm{Flip}(p,A)$. 
\begin{lemma}
\label{lm-coupling}
Given a set $A$ with $N$ elements and a number $m$, define $p^{\pm}(m)=\frac{m}{(1\mp \delta)N}$. Consider the random variables  $X\sim \mathrm{Samp}(m,A)$ and $Z^{\pm}\sim \mathrm{Flip}(p^{\pm},A)$, then  there exists a coupling $\mu$ such that: 
\begin{equation}
\P_{\mu}\left(Z^{-} \subseteq X \subseteq Z^{+} \right) \geq 1- 2 \exp\left(\frac{\delta^{2}}{3(1+\delta)m} \right).
\end{equation}
\end{lemma}
\begin{proof} Let $\mu$ be the joint distribution of $U_{1},\ldots,U_{N}$ i.i.d uniform in $[0,1]$ random variables, and $U_{(m)}$ to denote the $m$-th smallest such random variable. Define $X(U)=\{i\in A: U_{i}\leq U_{(m)}\}$ and $Z^{\pm}(U)=\{i\in A: U_{i} \leq p^{\pm}\}$ to be random subsets of $A$. By construction it is easy to see that $X(U)$ and $Z^{\pm}(U)$ have the right marginals. By construction of the sets, it is easy to see that the following equivalence holds:
\[
Z^{-} \subseteq X \subseteq Z^{+} \Leftrightarrow |Z^{-}| \leq |X| \leq |Z^{+}|
\]
To analyze the second event define the ``bad" events:
\begin{eqnarray*}
B_{-}&=&\{u \in [0,1]^{N}: \sum_{i=1}^{N}\mathbb{I}(u_{i}\leq p_{-}) > m
\} \\
B_{+}&=&\{u \in [0,1]^{N}: \sum_{i=1}^{N}\mathbb{I}(u_{i}\leq p_{+}) < m\}
\end{eqnarray*}
Each event can be stated as the probability that the sum $X_{\pm}$ of $n$ i.i.d Bernoulli $p_{\pm}$ random variables exceeds (smaller then) the expectation $np_{\pm}$. By employing standard Chernoff bounds, we get:
\begin{eqnarray*}
\P_{\mu}(B_{-})&=&\P_{U}(X_{-}>m) = \P_{\mu}(X>(1+\delta)np_{-}) \leq \exp(-\frac{\delta^{2}}{3(1+\delta)}m) \\
\P_{U}(B_{+})&=&\P_{U}(X_{+}<m) = \P_{U}(X<(1-\delta)np_{+}) \leq \exp(-\frac{\delta^{2}}{2(1-\delta)}m)\\
\end{eqnarray*}
The proof is concluded through the use of union bound: 
\[
\P_{\mu}(B_{-}\cup B_{+})\leq \P_{\mu}(B_{-})+\P_{U}(B_{+}) \leq 2\exp\left(-\frac{\delta^{2}}{3(1+\delta)}m\right)
\]
This concludes the lemma.
\end{proof}

Using, this simple lemma and Theorem 2, we  prove the \emph{sandwich theorem}.

\subsection{Main proof.}

Recall,  that our aim is to prove that the uniform measure over the set $S$ is $(\epsilon,\delta)$-sandwichable by some product measure $G(n,\mathbf{Q})$.

\begin{proof}[Proof of Theorem 1]
Given a $\mathcal{P}$-symmetric convex set $S$, consider $\mathbf{m}^{*}(S)$ the optimal edge-profile and define the $n\times n$ matrix $\mathbf{Q}^{*}(S)$ as: $Q_{u,v} = \frac{m^{*}_{i}}{p_{i}}, \forall \{u,v\}\in P_{i} \text{ and } i\in[k]$. Further, define $\mathbf{q}_{i}:=\frac{m^{*}_{i}}{p_{i}}, \forall \in [k]$ to be used later. In order to prove the required statement, we need to construct a coupling between the random variables $G\sim U(S)$, $G^{\pm}\sim G(n,(1\pm)\mathbf{Q}^{*})$.  By separating edges according to the partition, we express the edge set of the graphs as $E(G)=E_{1}\cup \ldots \cup E_{k}$ and $E(G^{\pm})= E_{1}^{\pm}\cup \ldots \cup E^{\pm}_{k}$. 

Let $\mu$ denote the joint probability distribution of $N+1$ i.i.d.\ uniform random variables $U_{1},\ldots, U_{N+1}$ on $[0,1]$. As in the coupling lemma, we are going to use these random variables to jointly generate the random edge-sets of $G^{-}, G, G^{+}$. Using $U_{N+1}$, we can first generate the edge profile $\mathbf{v}\in \mathbf{m}(S)$ from its corresponding distribution. Then, conditional on  the edge profile $\mathbf{v}\in \N^{k}$, the probability distribution of $G$ factorizes in $G(n,m)$-like distributions for each block (Section~\ref{app:basic}). Lastly, we associate with each edge $e$ a unique random variable $U_{e}$ and construct a coupling for edges in each block separately. 

In our notation, $E_{i}\sim \mathrm{Samp}(v_{i},P_{i})$ and $E^{\pm}_{i}\sim \mathrm{Flip}(q_{i}^{\pm},P_{i})$. Using Lemma\ref{lm-coupling},   we  construct a coupling for each $i\in [k]$ between the random variables $E_{i},E_{i}^{+},E_{i}^{-}$ and bound  the probability that the event $E_{i}^{-} \subseteq E_{i} \subseteq E_{i}^{+}$ does not hold.  Using the union bound over the $k$ parts, we then obtain an estimate of the probability that the property holds across blocks, always conditional on the edge-profile $\mathbf{v}$. The final step involves getting rid of the conditioning by invoking the concentration theorem.

Concretely, define $B_{i}$ the event that the $i$-th block does not satisfy the property  $E_{i}^{-} \subseteq E_{i} \subseteq E_{i}^{+}$ and $\mathcal{L}_{\epsilon}(\mathbf{m}^{*})$ the set appearing in Theorem \ref{thm-concentration}. We have that $\P_{\mu}(G^{-} \subseteq G \subseteq G^{+}) = 1- \P_{\mu}\left(\cup B_{i}\right)$.  Conditioning on the edge profile gives:
\begin{eqnarray*}
\P_{\mu}(\cup B_{i}) &\leq& \P_{\mu}(\mathcal{L}^{c}_{\epsilon}(\mathbf{m}^{*})) +  \sum_{\mathbf{v}\in L_{\epsilon}(\mathbf{m}^{*})} \P_{\mu}(\cup B_{i}|\mathbf{v})\P_{\mu}(\mathbf{v})\\ &\leq&  \P_{\mu}(\mathcal{L}^{c}_{\epsilon}(\mathbf{m}^{*})) +   \max_{\mathbf{v}\in L_{\epsilon}(\mathbf{m}^{*})}\P_{\mu}(\cup B_{i}|\mathbf{v})\\
&\leq&  \P_{\mu}(\mathcal{L}^{c}_{\epsilon}(\mathbf{m}^{*})) +   \max_{\mathbf{v}\in L_{\epsilon}(\mathbf{m}^{*})}\left[\sum_{i=1}^{k}\P_{\mu}(B_{i}|\mathbf{v})\right]
\end{eqnarray*}
The first inequality holds by conditioning on the edge profile and bounding the probability of the bad events from above by $1$ for all ``bad" profiles (outside of the concentration set). The second inequality, is derived by upper bounding the probability of the bad event by the most favorable such edge-profile and the last inequality follows from an application of the union bound. Applying Theorem \ref{thm-concentration}, we get a bound on the first term and then invoking Lemma \ref{lm-coupling} we get a bound for each of the term in the sum:
\[
\P_{\mu}(\cup B_{i}) \leq  \exp\left[-\mu(S)\left(\frac{\epsilon^{2}}{1+\epsilon} - \lambda(S)\right) \right] +  2 \max_{\mathbf{v}\in L_{\epsilon}(\mathbf{m}^{*})}\left[\sum_{i=1}^{k}\exp\left(-\frac{\epsilon^{2}}{3(1+\epsilon)}v_{i} \right)\right]
\]
Hence, we see that the upper bound is monotone in $v_{i}$ for all $i\in[k]$. Additionally, we know that for all $\mathbf{v}\in L_{\epsilon}(\mathbf{m}^{*})$ it holds that $\mathbf{v} \geq (1-\epsilon) \mathbf{m}^{*}$. Further, by definition we have $\mathbf{m}^{*}\geq \mu(S)$. The bound now becomes:
\begin{eqnarray*}
\P_{\mu}(\cup B_{i}) &\leq&  \exp\left[-\mu(S)\left(\frac{\epsilon^{2}}{1+\epsilon} - \lambda(S)\right) \right] +  2k\exp\left[-\frac{\epsilon^{2}(1-\epsilon)}{3(1+\epsilon)}\mu(S) \right]\\
&\leq&  \exp\left[-\mu(S)\left(\frac{\epsilon^{2}}{1+\epsilon} - \lambda(S)\right) \right] +\exp\left[-\mu(S)\left(\frac{\epsilon^{2}(1-\epsilon)}{3(1+\epsilon)}-\frac{\log (2k)}{\mu(S)} \right)\right]
\end{eqnarray*}
Finally, using $\epsilon<1/2$ and  $\log(2k)/\mu(S) \leq \lambda(S)$ we arrive at the required conclusion.
\end{proof}

\section{Supplementary Proofs}\label{app:basic}

\begin{proof}[Proof of Proposition~\ref{paparia}]
Fix an $\mathbf{x}\in H_{n}$ and consider the set $O(\mathbf{x})\stackrel{\vartriangle}{=}\{\mathbf{y}\in H_{n}: \exists \pi \in \Pi_{n}(\mathcal{P}) \text{ such that } y = \pi(\mathbf{x})\}$ and call it the \emph{orbit} of $\mathbf{x}$ under $\Pi_{n}(\mathcal{P})$(note that by group property orbits form a partition of $H_{n}$). The assumption of symmetry, implies that $f$ is constant for all $\mathbf{y}\in O(\mathbf{x})$: 
\[
f(\mathbf{y}_{1})=f(\mathbf{y}_{2})= f(\mathbf{x}), \ \forall \mathbf{y}_{1},\mathbf{y}_{2} \in O(\mathbf{x})  
\]
By  definition of $\Pi_{n}(\mathcal{P})$, for any $\mathbf{x}\in H_{n}$ there is a permutation $\pi_{x}\in \Pi_{n}(\mathcal{P})$, such that i) $\pi_{x}(\mathbf{x})=(\pi_{x,1}(x_{P_{1}}),\ldots,\pi_{x,k}(x_{P_{k}})) \in O(\mathbf{x})$,  ii) for all $i\in[k]$, $\pi_{x,i}(x_{P_{i}})$ is a bit-string where all $1$'s appear consequently starting from the first position. Let as identify with each orbit $O\subset H_{n}$ such a distinct element $\mathbf{x}_{o}$. As the function of $f$ is constant along each orbit, its value depends only through $\mathbf{x}_{o}$, which in turn depends only on the number of $1$'s(edges) in each \part, encoded in the edge profile $\mathbf{m}=(m_{1},\ldots,m_{k})$.
\end{proof}

\begin{proof}[Proof of Proposition~\ref{prop-lower}] Consider $\mathcal{P}$ any balanced partition consisting of $k$-parts and let $S=\mathcal{G}_{n}$ to be the space of all graphs. Then $Z(S)=|S|=2^{\binom{n}{2}}$ and $\mathbf{m}^{*}(S)$ is the all $\binom{n}{2}/2k$ vector (all blocks half full). Using Stirling's approximation of the factorial we have that:
\begin{eqnarray}
\log(|S|)-\text{\ENT}(\mathbf{m}^{*}) &\geq&  \binom{n}{2}\log 2 -2k\log\binom{\binom{n}{2}/k}{\frac{1}{2}\binom{n}{2}/k}\\
&\geq&  k \log n - \frac{k}{2}\log k
\end{eqnarray}
For $k=o(n^{2})$ the last expression is of order $\Omega(k \log(n))$.

\end{proof}

 \begin{proposition}\label{paparaki}
 Consider for all $i\in[k]$ disjoint sets of edges $I_{i},O_{i}\subset P_{i}$ and  define  the events $A_{i}=\{ G \in \mathcal{G}_{n}: I_{i}\subset E(G) \text{ and } O_{i}\cap E(G) = \emptyset \}$. Conditional on the edge profile of $G$ being $\mathbf{v}$, the events are independent, i.e. it holds that: $\P_{S}\left(A_{1}\cap \ldots \cap A_{k}|\mathbf{v}\right) = \prod_{i=1}^{k}\P_{S}(A_{i}|v_{i})$.
 \end{proposition}

\begin{proof}[Proof of Proposition~\ref{paparaki}]
Since $G\sim U(S)$ the distribution of $G$ is by definition uniform on $S$. This also means that it is uniform on the subset of graphs having edge profile $\mathbf{v}\in \nats^{k}$(conditioning). But then:
\begin{eqnarray*}
\P_{S}\left(A_{1}\cap \ldots \cap A_{k}|\mathbf{v}\right) &=& \frac{\P_{S}\left(A_{1}\cap \ldots \cap A_{k}\cap \mathbf{m}(G)=\mathbf{v}\right)}{\P_{S}(\mathbf{m}(G)=\mathbf{v})} \\
&=& \frac{\left|A_{1}\cap \ldots \cap A_{k}\cap \mathbf{m}(G)=\mathbf{v} \right|}{\left| \mathbf{m}(G)=\mathbf{v} \right|}\mathbb{I}_{\mathbf{m}(S)}(\mathbf{v})
\end{eqnarray*}
where the first equality follows from Bayes rule and the second due to uniformity and the fact that our symmetry assumption implies that membership in $S$ depends only on the edge-profile $\mathbf{m}$. Recall that each set $A_{i}=\{ G \in \mathcal{G}_{n}: I_{i}\subset E(G) \text{ and } O_{i}\cap E(G) = \emptyset \}$ imposes the requirement that the edges in $I_{i}$ are included in $G$ and that the edges in $O_{i}$ are not included in $G$. Having conditioned on $\mathcal{v}$, we know that exactly $v_{i}$ edges from $P_{i}$ are included in $G$ and that we can satisfy the requirements for edges in $P_{i}$ by selecting any subset of $v_{i}-|I|_{i}$ edges out of $P_{i}\setminus\left(I_{i}\cup _{i} \right)$. For convenience set $|P_{i}|=p_{i}$, $|I|_{i}=n_{i}$, $|O_{i}\cup I_{i}|=r_{i}$, and let $C^{n}_{\ell}$ denote the number of $k$-combinations out of an $n$ element set(binomial coefficient). The number of valid subsets of $P_{i}$ is then given by $C^{p_{i}-r_{i}}_{ v_{i}-n_{i}}$. As the constraints imposed are separable, we have:
\begin{eqnarray*}
\frac{\left|A_{1}\cap \ldots \cap A_{k}\cap \mathbf{m}(G)=\mathbf{v} \right|}{\left| \mathbf{m}(G)=\mathbf{v} \right|}= \frac{\prod_{i=1}^{k}C^{p_{i}-r_{i}}_{ v_{i}-n_{i}}}{\left|\mathbf{m}(G)=\mathbf{v} \right|} = \prod_{i=1}^{k}  \frac{|A_{i}\cap \mathbf{m}(G)=\mathbf{v}|}{\left|\mathbf{m}(G)=\mathbf{v} \right|}
\end{eqnarray*}
which gives the required identity by exploiting again uniformity of the probability  measure.
\end{proof}	
\bibliographystyle{plain}
\bibliography{fixed}

\begin{thebibliography}{10}

\bibitem{navigating}
D.~{Achlioptas} and P.~{Siminelakis}.
\newblock {Navigability is a Robust Property}.
\newblock {\em ArXiv}, January 2015.

\bibitem{bolo_book}
B{\'e}la Bollob{\'a}s.
\newblock {\em Random graphs}, volume~73 of {\em Cambridge Studies in Advanced
  Mathematics}.
\newblock Cambridge University Press, Cambridge, second edition, 2001.

\bibitem{borgs2014p}
C.~{Borgs}, J.~T. {Chayes}, H.~{Cohn}, and Y.~{Zhao}.
\newblock {An {$L^{p}$} theory of sparse graph convergence I: limits, sparse
  random graph models, and power law distributions}.
\newblock {\em ArXiv e-prints}, January 2014.

\bibitem{chung2006complex}
Fan~RK Chung and Linyuan Lu.
\newblock {\em Complex graphs and networks}, volume 107.
\newblock American mathematical society Providence, 2006.

\bibitem{durrett}
Richard Durrett.
\newblock {\em Random graph dynamics}, volume~20.
\newblock Cambridge university press, 2007.

\bibitem{erdHos1959}
Paul Erd{\H{o}}s and Alfr{\'e}d R{\'e}nyi.
\newblock On random graphs.
\newblock {\em Publicationes Mathematicae Debrecen}, 6:290--297, 1959.

\bibitem{frieze1999quick}
Alan Frieze and Ravi Kannan.
\newblock Quick approximation to matrices and applications.
\newblock {\em Combinatorica}, 19(2):175--220, 1999.

\bibitem{gilbert}
Edgar~N Gilbert.
\newblock Random graphs.
\newblock {\em The Annals of Mathematical Statistics}, pages 1141--1144, 1959.

\bibitem{Goldenberg}
Anna Goldenberg, Alice~X. Zheng, Stephen~E. Fienberg, and Edoardo~M. Airoldi.
\newblock A survey of statistical network models.
\newblock {\em Found. Trends Mach. Learn.}, 2(2):129--233, February 2010.

\bibitem{hoff2002latent}
Peter~D Hoff, Adrian~E Raftery, and Mark~S Handcock.
\newblock Latent space approaches to social network analysis.
\newblock {\em Journal of the american Statistical association},
  97(460):1090--1098, 2002.

\bibitem{Janson_book}
Svante Janson, Tomasz {\L}uczak, and Andrzej Rucinski.
\newblock {\em Random graphs}.
\newblock Wiley-Interscience Series in Discrete Mathematics and Optimization.
  Wiley-Interscience, New York, 2000.

\bibitem{Nature}
Jon~M. Kleinberg.
\newblock {Navigation in a small world}.
\newblock {\em Nature}, 406(6798):845, August 2000.

\bibitem{NIPS}
Jon~M. Kleinberg.
\newblock Small-world phenomena and the dynamics of information.
\newblock In Thomas~G. Dietterich, Suzanna Becker, and Zoubin Ghahramani,
  editors, {\em NIPS}, pages 431--438. MIT Press, 2001.

\bibitem{leskovec2010kronecker}
Jure Leskovec, Deepayan Chakrabarti, Jon Kleinberg, Christos Faloutsos, and
  Zoubin Ghahramani.
\newblock Kronecker graphs: An approach to modeling networks.
\newblock {\em The Journal of Machine Learning Research}, 11:985--1042, 2010.

\bibitem{distortion}
Nathan Linial, Eran London, and Yuri Rabinovich.
\newblock The geometry of graphs and some of its algorithmic applications.
\newblock {\em Combinatorica}, 15(2):215--245, 1995.

\bibitem{Lovasz2012}
L{\'a}szl{\'o} Lov{\'a}sz.
\newblock {\em Large networks and graph limits}, volume~60.
\newblock American Mathematical Soc., 2012.

\bibitem{mode1971multitype}
C.J. Mode.
\newblock {\em Multitype branching processes: theory and applications}.
\newblock Modern analytic and computational methods in science and mathematics.
  American Elsevier Pub. Co., 1971.

\bibitem{newman2010networks}
Mark Newman.
\newblock {\em Networks: an introduction}.
\newblock Oxford University Press, 2010.

\end{thebibliography}

\end{document}